\documentclass{article}
\usepackage[a4paper]{geometry}
\usepackage{xcolor}
\usepackage{amssymb}
\usepackage{amsmath}
\usepackage{amsthm}
\usepackage[mathscr]{euscript}
\usepackage{bm}
\usepackage{cite}
\usepackage{tensor}

\numberwithin{equation}{section}

\theoremstyle{plain}
\newtheorem{thm}{Theorem}[section]

\newtheorem{lem}[thm]{Lemma}

\theoremstyle{definition}
\newtheorem{defn}[thm]{Definition}
\newtheorem{exm}[thm]{Example}
\newtheorem*{defn2}{Ranking heuristic}
\newtheorem{rout}{Procedure}

\newcommand{\beq}{\begin{equation}}
\newcommand{\eeq}{\end{equation}}
\newcommand{\bex}{\begin{exm}}
\newcommand{\eex}{\end{exm}}

\newcommand{\mbu}{\mathbf{u}}

\newcommand{\mbx}{\mathbf{x}}

\newcommand{\mbE}{\mathbf{E}}
\newcommand{\mbF}{\mathbf{F}}
\newcommand{\mbI}{\mathbf{I}}
\newcommand{\mbJ}{\mathbf{J}}
\newcommand{\mbK}{\mathbf{K}}

\newcommand{\mbone}{\boldsymbol{1}}
\newcommand{\mcA}{\mathcal{A}}

\newcommand{\mcC}{\mathcal{C}}
\newcommand{\mcD}{\mathcal{D}}

\newcommand{\mcQ}{\mathcal{Q}}

\newcommand{\upd}{\,\mathrm{d}}
\newcommand{\ua}{u^\alpha}
\newcommand{\uaJ}{u^\alpha_{\mathbf{J}}}
\newcommand{\uaI}{u^\alpha_{\mathbf{I}}}
\newcommand{\uaIj}{u^\alpha_{\mathbf{I},j}}
\newcommand{\uaIk}{u^\alpha_{\mathbf{I},k}}
\newcommand{\uaIjk}{u^\alpha_{\mathbf{I},j+k}}
\newcommand{\Ea}{\mbE_{u^\alpha}}
\newcommand{\Ex}{\mbE_{u}^{x}}

\newcommand{\EaIx}{\mbE_{u^\alpha_\mathbf{I}}^{x}}

\newcommand{\EaIkx}{\mbE_{u^\alpha_{\mathbf{I},k}}^{x}}
\newcommand{\sigx}{\sigma_u^x}
\newcommand{\sigaIx}{\sigma_{u^\alpha_\mathbf{I}}^{x}}
\newcommand{\pix}{\pi_u^x}
\newcommand{\piaIx}{\pi_{u^\alpha_\mathbf{I}}^{x}}
\newcommand{\bdag}{\boldsymbol{\dag}}

\newcommand{\p}{\partial}

\title{Partial Euler operators and the efficient inversion of Div}

\author{P. E. Hydon\\
	School of Mathematics, Statistics and Actuarial Science,\\
	University of Kent, Canterbury CT2 7FS, UK}

\date{15 December 2022}

\begin{document}

\label{firstpage}
\maketitle

\begin{abstract}
The problem of inverting the total divergence operator is central to finding components of a given conservation law. This might not be taxing for a low-order conservation law of a scalar partial differential equation, but integrable systems have conservation laws of arbitrarily high order that must be found with the aid of computer algebra. Even low-order conservation laws of complex systems can be hard to find and invert. This paper describes a new, efficient approach to the inversion problem. Two main tools are developed: partial Euler operators and partial scalings. These lead to a line integral formula for the inversion of a total derivative and a procedure for inverting a given total divergence concisely. 
\end{abstract}

\section{Introduction}

Around 20 years ago, Stephen Anco and George Bluman \cite{AncBlu1,AncBlu2} introduced a comprehensive practical method for determining conservation laws of partial differential equations (PDEs) in Kovalevskaya form. The method is based on finding adjoint symmetries and applying Helmholtz conditions\footnote{The deep theoretical foundation for this approach is discussed in Olver \cite{Olver}. For a recent review of the method and its extension beyond equations in Kovalevskaya form, see Anco \cite{Anco}. In particular, the method may be used for equations in extended Kovalevskaya form (see Popovych \& Bihlo \cite{PB}).}. A key part of the calculation is the inversion of the total divergence operator $\mathrm{Div}$ to obtain the components of the conservation law. Usually, this can be done by using a homotopy operator, but the following three problems may occur with the standard homotopy formula (which is given by Olver in \cite{Olver}).
\begin{enumerate}
	\item The homotopy formula uses definite integrals. This works well if the divergence is a differential polynomial; by contrast, rational polynomials commonly have a singularity at one limit. Hickman \cite{Hickman} and Poole \& Hereman \cite{Poole10} suggest avoiding these by working in terms of indefinite integrals, an approach that we use throughout this paper. Alternatively, one can move the singularity by modifying the dependent variable (see Anco \& Bluman \cite{AncBlu2} and Poole \& Hereman \cite{Poole10}).
	\item Scaling is fundamental to the homotopy approach to inversion. For instance, varying the scaling parameter in the standard homotopy formula moves contributions from variables along a ray to the origin. However, scaling does not change rational polynomials that are homogeneous of degree zero, so the standard inversion process does not work for such terms. Deconinck \& Nivala \cite{DecNiv} discussed this problem in some detail (for one independent variable only) and suggested possible workarounds, but commented, `We are unaware of a homotopy method that algorithmically avoids all problems like the ones demonstrated ...'. Poole \& Hereman \cite{Poole10} proposed an approach that works well for problems with one independent variable, but noted the difficulties of extending this to multiple independent variables (in a way that can be programmed).
	\item The standard homotopy operator applies to a starshaped domain and integrates along rays to the origin, changing all Cartesian coordinates at once. This gives an inefficient inversion, in the sense that the number of terms is generally very much greater than necessary; the homotopy formula creates superfluous elements of $\mathrm{ker}(\mathrm{Div})$. For polynomial divergences, Poole \& Hereman \cite{Poole10} removed curls by parametrizing all terms given by the homotopy formula and optimizing the resulting linear system. This approach is very effective, because (except where there are cancellations), the homotopy formula tends to include every possible term that can appear in an inversion. However, it is unclear whether this approach can be generalized to non-polynomial divergences. Moreover, the removal of curl terms takes extra processing time and does not allow for the possibility that new terms might appear in the most concise inversion of the divergence. If inversion could be done with respect to one independent variable at a time, this might prevent the occurrence of superfluous curls from the outset.
\end{enumerate}

\bex
To illustrate the inefficiency of homotopy operators on starshaped domains, consider the following divergence in $\mathbb{R}^3$:
\[
\mcC(x,y,z)=2xy\cos z.
\]
The homotopy operator based on a starshaped domain gives $\mcC=\mathrm{div}(x\phi,y\phi,z\phi)$, where
\begin{align*}
\phi&=\int_{0}^1\lambda^2\mcC(\lambda x,\lambda y,\lambda z)\,\upd\lambda\\
&=2xy\left\{\left(z^{-1}-12z^{-3}+24z^{-5}\right)\sin z +\left(4z^{-2}-24z^{-4}\right)\cos z\right\}.
\end{align*}
By comparison, for a given divergence $\mcC(x,y,z)$ that has no singularities on the coordinate axes, using a homotopy formula that integrates one variable at a time gives $\mcC=\mathrm{div}(F,G,H)$, where
\[
F=\int_0^x\mcC(\lambda,y,z)\upd\lambda,\qquad G=\int_0^y\mcC(0,\lambda,z)\upd\lambda,\qquad H=\int_0^z\mcC(0,0,\lambda)\upd\lambda.
\]
This recovers the concise form $(F,G,H)=(x^2y\cos z,0,0)$. However, the use of the lower limit makes the formula over-elaborate (and unsuitable for treating singularities). Indefinite integration is far more straightforward:
\[
\mcC=\frac{\p F}{\p x}, \quad \text{where}\quad F=\int \mcC(x,y,z) \upd x=x^2y\cos z.
\]
So the homotopy formula for starshaped domains gives $14$ more terms than the simple form above; the superfluous terms amount to
\[
\mathrm{curl}\left(4xy^2\left\{\left(3z^{-2}-6z^{-4}\right)\sin z -\left(z^{-1}-6z^{-3}\right)\cos z\right\},\,\,x^2y\sin z,\,\, xy\phi/2\right).
\]
\eex

The current paper extends the efficient one-variable-at-a-time approach to total derivatives. Indefinite integration is used, as advocated by Hickman \cite{Hickman} for standard homotopy operators; it avoids the complications resulting from singularities. From the computational viewpoint, the biggest advantage of integration with respect to just one independent variable is that the major computer algebra systems have efficient procedures for computing antiderivatives.

The keys to inverting a total divergence one variable at a time are `partial Euler operators' (see Section \ref{pEul}). These enable the inversion of a total derivative $D_x$ to be written as an indefinite line integral. Section \ref{DivInversion} introduces a new iterative method for inverting a given total divergence; typically, this does not produce superfluous terms and very few iterations are needed. Furthermore, it can cope with components that are unchanged by the relevant scaling operator.

The methods in this paper are systematic, leading to procedures that are intended to be suitable for implementation in computer algebra systems.

\section{Standard differential and homotopy operators}\label{Basics}

Here is a brief summary of the standard operators that are relevant to total divergences; for further details, see Olver \cite{Olver}. The independent variables $\mathbf{x} = (x^1,\dots, x^p)$ are local Cartesian coordinates and the dependent variables $\mathbf{u}=(u^1,\dots,u^q)$ may be real- or complex-valued. The Einstein summation convention is used to explain the main ideas and state general results. In examples, commonly-used notation is used where this aids clarity. All functions are assumed to be locally smooth, to allow the key ideas to be presented simply. 

Derivatives of each $\ua$ are written as $\uaJ$, where $\mbJ=(j^1,\dots,j^p)$ is a multi-index; each $j^i$ denotes the number of derivatives with respect to $x^i$, so $\ua_{\mathbf{0}}= \ua$. The variables $x^i$ and $\uaJ$ can be regarded as jet space coordinates. The \textit{total derivative} with respect to $x^i$,
\[
D_i=\frac{\p}{\p x^i}+\ua_{\mbJ +\mbone_i}\,\frac{\p}{\p \uaJ}\,,\quad\text{where}\quad \mbJ +\mbone_i=(j^1,\dots,j^{i-1},j^i+1,j^{i+1},\dots, j^p),
\]
treats each $\uaJ$ as a function of $\mbx$. To keep the notation concise, write
\[
D_\mbJ=D_1^{j^1}D_2^{j^2}\cdots D_p^{j^p};
\]
note that $\uaJ=D_\mbJ(\ua)$. Let $[\mbu]$ represent $\mbu$ and finitely many of its derivatives; more generally, square brackets around an expression denote the expression and as many of its total derivatives as are needed.

A \textit{total divergence} is an expression of the form
\[
\mcC=\mathrm{Div}(\mbF):= D_iF^i(\mbx,[\mbu]).
\]
(If all $F^i$ depend on $\mbx$ only, $\mcC$ is an ordinary divergence.) A conservation law of a given system of partial differential equations (PDEs), $\mcA_\ell(\mbx,[\mbu])=0,\ \ell=1,\dots,L$, is a total divergence that is zero on all solutions of the system; each $F^i$ is a finite sum of terms. By using elementary algebraic operations (in particular, expanding logarithms of products and products of sums), the number of linearly independent terms may be maximized. When the number of linearly independent terms is maximal for each $i$, we call the result the \textit{fully-expanded} form of $\mbF$. 

When $p>1$, the $p$-tuple of components, $\mbF$, is determined by $\mcC$ up to a transformation of the form
\beq\label{trans}
F^i\longmapsto F^i+D_j\left\{f^{ij}(\mbx,[\mbu])\right\},\qquad\text{where}\  f^{ji}=-f^{ij}.
\eeq
(If $p=3$, such a transformation adds a total curl to $\mbF$.) The total number of terms in $\mbF$ is the sum of the number of terms in all of the fully-expanded components $F^i$. If this cannot be lowered by any transformation \eqref{trans}, we call $\mbF$ \textit{minimal}. Commonly, there is more than one minimal $\mbF$, any of which puts the inversion of $\mathrm{Div}$ in as concise a form as possible. If $p=1$, the sole component $F^1$ (also denoted $F$) is determined up to an arbitrary constant, so the number of non-constant terms is fixed.

The formal adjoint of a differential operator (with total derivatives), $\mcD$, is the unique differential operator $\mcD^{\bdag}$ such that
\[
f\,\mcD g - \left(\mcD^{\bdag} f\right)g
\]
is a total divergence for all functions $f(\mbx,[\mbu])$ and $g(\mbx,[\mbu])$. In particular,
\beq\label{modJ}
(D_\mbJ)^{\bdag}=(-D)_\mbJ:=(-1)^{|\mbJ|}D_\mbJ,\quad\text{where}\quad |\mbJ|=j^1+\cdots+j^p.
\eeq
Thus the (standard) Euler--Lagrange operator corresponding to variations in $\ua$ is
\[
\Ea=(-D)_\mbJ\frac{\p}{\p \uaJ}\,.
\]
Total divergences satisfy a useful identity: a function $\mcC(\mbx,[\mbu])$ is a total divergence if and only if
\beq\label{divcrit}
\Ea(\mcC)=0,\qquad \alpha=1,\dots,q.
\eeq

Given a Lagrangian function $L(\mbx,[\mbu])$, the Euler--Lagrange equations are $\Ea(L)=0$. Given a set of Euler--Lagrange equations that are polynomial in the variables $(\mbx,[\mbu])$, the function $\overline{L}$ given by the homotopy formula
\beq
\overline{L}(\mbx,[\mbu])=\int_{0}^1\ua\left\{\Ea(L)\right\}\!\big|_{[\mbu\mapsto\lambda\mbu]}\upd\lambda
\eeq
differs from $L$ by a total divergence. (The same applies to many, but not all, non-polynomial Euler--Lagrange equations.)

When $p=1$, the equation $P(x,[\mbu])=D_xF$ is invertible (at least, for polynomial $P$) by the following standard homotopy formula:
\beq\label{homDxF}
F(x,[\mbu])=\!\int_{0}^1\sum_{i=1}^\infty D_x^{i-1}\!\left(\ua\!\left\{\sum_{k\geq i}\binom{k}{i}(-D_x)^{k-i}\frac{\p P(x,[\mbu])}{\p (D_x^k\ua)}\right\}\Bigg|_{[\mbu\mapsto\lambda\mbu]}\right)\!\upd\lambda +\int_{0}^1\!xP(\lambda x,[0])\upd\lambda.
\eeq
The operator acting on $P$ in the braces above is the \textit{higher Euler operator} of order $i$ for $p=1$. When $p\geq 2$, the standard homotopy formula is similar, but somewhat more complex (see Olver \cite{Olver} for details); it is based on higher Euler operators and integration along rays in a totally starshaped domain. The following example illustrates that even for quite simple divergences, this formula commonly yields inversions with many superfluous terms.

\bex\label{BBM}
The Benjamin--Bona--Mahony (BBM) equation, $u_{t}\!-\!uu_x\!-\!u_{xxt}=0$, has a conservation law
\beq\label{BBMeq}
\mcC=D_xF+D_tG=(u^2\!+\!2u_{xt})(u_{t}\!-\!uu_x\!-\!u_{xxt}).
\eeq
The standard homotopy formula gives
\begin{align*}
F&=-\tfrac{1}{3}uu_{xxtt}+\tfrac{2}{3}u_{x}u_{xtt}-\tfrac{1}{3}u_{t}u_{xxt}-\tfrac{1}{3}u_{xt}^2-\tfrac{1}{2}uu_{tt}+\tfrac{1}{2}u_{t}^2-\tfrac{2}{3}u^2u_{xt}+\tfrac{2}{3}uu_xu_{t}-\tfrac{1}{4}u^4,\\
G&=\tfrac{1}{3}uu_{xxxt}-\tfrac{1}{3}u_{x}u_{xxt}-\tfrac{2}{3}u_{xx}u_{xt}+\tfrac{1}{2}uu_{xt}+\tfrac{1}{2}u_xu_{t}-\tfrac{1}{3}u^2u_{xx}-\tfrac{2}{3}uu_{x}^2+\tfrac{1}{3}u^3,
\end{align*}
a total of 17 terms. By contrast,  careful integration by inspection yields
\beq\label{BBMCL}
\mcC=D_x\left(u_t^2-u_{xt}^2-u^2u_{xt}-\tfrac{1}{4}u^4\right)+D_t\left(\tfrac{1}{3}u^3\right),
\eeq
which is minimal, having only five terms in the components.
\eex

The homotopy formulae above can be applied or adapted to some, but not all, classes of non-polynomial Lagrangians and divergences. %(see Section \ref{Homotopy}).

\section{Partial Euler operators and partial scalings}\label{pEul}

This section introduces some ideas and results that underpin integration with respect to one independent variable at a time.
The independent variable over which one integrates is distinguished; this is denoted by $x$. For instance, if $x=x^1$, replace the derivative index $\mbJ$ by $(\mbI,j)$, where $j=j^1$ and $\mbI=(j^2,\dots,j^{p})$. So the dependent variables and their derivatives are denoted
\[
\ua_{\mbI,\,j}=D_x^{\,j}\uaI,\qquad\text{where}\ \uaI=\ua_{\mbI,0}.
\]
In examples, however, we write each $\uaI$ more simply (as $u$, $v_y$, $u_{yt}$, and so on), using $\phantom{}_{,j}$ for $D_x^{\,j}$.

\subsection{Partial Euler operators}

The \textit{partial Euler operator} with respect to $x$ and $\uaI$ is obtained by varying each $\uaI$ independently, treating $x$ as the sole independent variable:
\beq
\EaIx=(-D_x)^j\frac{\p}{\p \uaIj}\,.
\eeq
Consequently, the standard Euler operator with respect to $\ua$ amounts to
\beq\label{Eulp}
\Ea=(-D)_\mbI\EaIx\,.
\eeq
Similarly, the partial Euler operator with respect to $x$ and $\uaIk$ is
\beq
\EaIkx=(-D_x)^j\frac{\p}{\p \uaIjk}\,.
\eeq
Note that
\beq\label{uppE}
\EaIkx=\frac{\p}{\p \uaIk}-D_x\mbE^x_{\ua_{\mbI,k+1}}.
\eeq
The following identities are easily verified; here, $f(\mbx,[\mbu])$ is an arbitrary function.
\begin{align}
	&\EaIx(D_x f)=0,\label{id1}\\
	&\EaIkx(D_x f)=\frac{\p f}{\p \ua_{\mbI,k-1}}\,,\qquad k\geq 1,\label{id2}\\
	&\EaIkx(D_i f)=D_i\left(\EaIkx(f)\right)+\mbE^x_{\ua_{\mbI-\mbone_i,k}}(f),\qquad x^i\neq x,\label{id3}
\end{align}
where the last term in \eqref{id3} is zero if $j^i=0$.

\subsection{Inversion of $D_x$}

The identities \eqref{id1}, \eqref{id2} and \eqref{id3} are central to the inversion of total divergences, including the following inversion of $P=D_xF$ as an indefinite line integral.

\begin{lem}
	If $P(\mbx,[\mbu])=D_xF$, then, up to an (irrelevant) arbitrary function of all independent variables other than $x$,
	\beq\label{line}
	F(\mbx,[\mbu])=\int \left(P-\sum_{k\geq 1}\uaIk\,\EaIkx(P)\right)\!\upd x+\sum_{k\geq 0}\mbE^x_{\ua_{\mbI,k+1}}(P)\upd\uaIk\,.
	\eeq
\end{lem}
\begin{proof}
	By the identity \eqref{id2},
	\[
	F(\mbx,[\mbu])=\int \frac{\p F}{\p x}\,\upd x+\sum_{k\geq 0}\mbE^x_{\ua_{\mbI,k+1}}(D_xF)\upd\uaIk\,.
	\]
	Moreover,
	\[
		P=\frac{\p F}{\p x}+\sum_{k\geq 1}\uaIk\,\frac{\p F}{\p \ua_{\mbI,k-1}}=\frac{\p F}{\p x}+\sum_{k\geq 1}\uaIk\,\EaIkx(D_x F)\,.
	\]
	Substituting $P$ for $D_xF$ completes the proof.
\end{proof}

\bex
Locally, away from its singularities, the function
\beq\label{awkDF}
P=\frac{u_{xx}v_y-u_xv_{xy}}{v_y^2}+\frac{uv_x-u_xv}{v(u+v)}+\frac{1}{x}
\eeq
belongs to $\mathrm{im}(D_x)$, but cannot be inverted using the standard homotopy formula. Substituting
\begin{align*}
&\mbE^x_{u_{,1}}(P)=\frac{\p P}{\p u_x}-D_x\frac{\p P}{\p u_{xx}}=-\frac{1}{u+v}\,,
&&\mbE^x_{u_{,2}}(P)=\frac{\p P}{\p u_{xx}}=\frac{1}{v_y}\,,\\
&\mbE^x_{v_{,1}}(P)=\frac{\p P}{\p v_{x}}=\frac{u}{v(u+v)}\,,
&&\mbE^x_{v_{y,1}}(P)=\frac{\p P}{\p v_{xy}}=-\frac{u_x}{v_y^2}\,,
\end{align*}
into \eqref{line} yields the inversion:
\beq\label{awkDFsol}
F=\int \frac{\upd x}{x} -\frac{\upd u}{u+v}\,+\frac{\upd u_x}{v_y} +\frac{u\upd v}{v(u+v)}-\frac{u_x\upd v_y}{v_y^2}\, =\, \ln\left|\frac{xv}{u+v}\right|+\frac{u_x}{v_y}\,.
\eeq
\eex

\subsection{Integration by parts}

From here on, we will restrict attention to total divergences $\mcC$ whose fully-expanded form has no terms that depend on $\mbx$ only. Such terms can be inverted easily by evaluating an indefinite integral, as explained in the Introduction. Henceforth, all indefinite integrals denote antiderivatives with the minimal number of terms in their fully-expanded form. Any arbitrary constants and functions that would increase the number of terms are set to zero. This restriction facilitates the search for minimal inversions.

The indefinite line integral formula \eqref{line} is closely related to integration by parts. To see this, we introduce a \textit{positive ranking} on the variables $u^\alpha_{\mbJ}\,$; this is a total order $\preceq$ that is subject to two conditions:
\[
(i)\,\ u^\alpha\prec u^\alpha_{\mbJ}\,,\quad\mbJ\neq\mathbf{0},\qquad (ii)\,\ u^\beta_{\mbI}\prec u^\alpha_{\mbJ}\Longrightarrow D_{\mbK }u^\beta_{\mbI}\prec D_{\mbK }u^\alpha_{\mbJ}\,.
\]
The \textit{leading part} of a differential function is the sum of terms in the function that depend on the highest-ranked $\uaJ$, and the \textit{rank} of the function is the rank of its leading part (see Rust \textit{et al.} \cite{Rust} for details and references). Let $f(\mbx,[\mbu])$ denote the leading part of the fully-expanded form of $F$ and let $\mathrm{U}_{,k}$ denote the highest-ranked $\uaIk$; then the highest-ranked part of $P=D_xF$ is $\mathrm{U}_{,k+1}\p f/\p \mathrm{U}_{,k}$. Then \eqref{line} includes the contribution
\[
\int \mbE^x_{\mathrm{U}_{,k+1}}(P)\upd \mathrm{U}_{,k}=\int \frac{\p f}{\p \mathrm{U}_{,k}}\,\upd \mathrm{U}_{,k}= f +\ \text{lower-ranked terms}.
\]
Integration by parts gives the same result. Subtracting $f$ from $F$ and iterating shows that evaluating the line integral \eqref{line} is equivalent to integrating by parts from the highest-ranked terms downwards.

Integration by parts is useful for splitting a differential expression $P(\mbx,[\mbu])$, with $P(\mbx,[0])=0$, into $D_xF$ and a remainder, $R$, whose $x$-derivatives are of the lowest-possible order. The splitting is achieved by the following procedure.
\begin{rout}\label{ibp} \textbf{Integration by parts}  
\begin{enumerate}
	\item[] 
\begin{enumerate}
	\item[]
\begin{enumerate}
	\item[\textsc{Step} 0.] Choose a positive ranking in which $\ua_{\mbI,0}\prec D_xu^\beta$ for all $\alpha, \mbI$ and $\beta$. (We call such rankings $x$-\textit{dominant}.) Initialize by setting $F:=0$ and $R:=0$.
	\item[\textsc{Step} 1.] Identify the highest-ranked $\uaIk$ in $P$; denote this by $\mathrm{U}_{,k}$. If $k=0$, add $P$ to $R$ and \textbf{stop}. Otherwise, determine the leading part, $g$, of $P$.
	\item[\textsc{Step} 2.] Determine the sum $h\mathrm{U}_{,k}$ of all terms in the fully-expanded form of $g$ that are of the form $\gamma\mathrm{U}_{,k}$, where $\gamma$ is ranked no higher than $\mathrm{U}_{,k-1}$, and let
	\[
	H=\int h\upd\mathrm{U}_{,k-1}\,.
	\]
	\item[\textsc{Step} 3.] Update $F, R$ and $P$, as follows:
	\[
	F:=F+H,\qquad R:=R+g-h\mathrm{U}_{,k},\qquad P:=P-g+h\mathrm{U}_{,k}-D_xH.
	\]
	If $P\neq 0$, return to \textsc{Step} 1. Otherwise output $F$ and $R$, then \textbf{stop}.
\end{enumerate}
\end{enumerate}
\end{enumerate}
\end{rout}

The reason for choosing an $x$-dominant ranking is to ensure that the derivative order with respect to $x$ outweighs all other ranking criteria. Consequently, the minimally-ranked remainder cannot contain $x$-derivatives of unnecessarily high order.

\bex\label{HDsimp}
To produce a concise inversion of a conservation law of the Harry Dym equation (see Example \ref{HD} below), it is necessary to split
\[
P=-\tfrac{8}{3}u^2u_{,4}-\tfrac{16}{3}uu_{,1}u_{,3}-4uu_{,2}^2+4u_{,1}^2u_{,2}-u^{-1}u_{,1}^4\,.
\]
%\[
%P=-8u_{,5}u_{,1}-10u_{,3}u_{,1}^3-10u_{,2}^2u_{,1}^2-u_{,1}^6\,.
%\]
Procedure \ref{ibp} gives the splitting
\[
P=D_x\!\left\{-\tfrac{8}{3}u^2u_{,3}+\tfrac{4}{3}u_{,1}^3\right\}-4uu_{,2}^2-u^{-1}u_{,1}^4\,.
\]
\eex

\bex The ranking criterion in \textsc{Step} 2 of Procedure \ref{ibp} ensures that there are no infinite loops. It is not enough that terms are linear in the highest-ranked $x$-derivative, as shown by the following splitting of
\[
P=\frac{v_{,3}}{u_{y}}+\frac{u_{,2}}{v_{y}}
\]
For the positive $x$-dominant ranking defined by $v\prec u\prec v_y$, Procedure \ref{ibp} yields
\[
P=D_x\!\left\{\frac{v_{,2}}{u_{y}}+\frac{u_{,1}}{v_{y}}\right\}+\frac{v_{,2}u_{y,1}}{u_{y}^2}+\frac{u_{,1}v_{y,1}}{v_{y}^2}\,.
\]
Both terms in the remainder are linear in their highest-ranked $x$-derivatives, which are $v_{,2}$ and $v_{y,1}$ respectively. However, further integration by parts would return $P$ to a form with a higher-ranked remainder.
\eex

\subsection{Partial scalings} 

To investigate partial Euler operators further, it is helpful to use a variant of the homotopy approach. The \textit{partial scaling} (by a positive real parameter, $\lambda$) of a function $f(\mbx,[\mbu])$ with respect to $x$ and $\uaI$ is the mapping
\[
\sigaIx:(f;\lambda)\mapsto f\big|_{\{\uaIj\mapsto \lambda \uaIj,\ j\geq 0\}}.
\]
Again, each $\uaI$ is treated as a distinct dependent variable. Note the identity
\beq\label{scalder}
	\sigaIx D_x= D_x \sigaIx\,.
\eeq

\begin{defn}
	The partial scaling $\sigaIx$ is a \emph{good scaling} for a given differential function $f(\mbx,[\mbu])$ if
	\beq\label{goods1}
	\sigaIx(f;\lambda)=\int \frac{\upd}{\upd \lambda}\!\left(\sigaIx(f;\lambda)\right)\upd \lambda,
	\eeq
	for all $\lambda$ in some neighbourhood of $1$.
\end{defn}

By definition, the partial scaling $\sigaIx$ fails to be a good scaling for $f$ if and only if there are terms that are independent of $\lambda$ in the fully-expanded form of $\sigaIx(f;\lambda)$. The simplest cause of this is that the fully-expanded form of $f$ has terms that are independent of $\uaI$ and its $x$-derivatives. However, this is not the only cause, as the following example illustrates. 

\bex\label{kerpi}
The scalings $\sigma^y_{u}$ and $\sigma^y_v$ are not good scalings for
\beq\label{awkDiv}
\mcC=u_{x}(2u+v_{y})-v_{x}(u_y+2v_{yy})+\frac{u_{x}}{u^2}+\frac{v_{yy}}{v_{y}}+\frac{2u_{y}}{u}\ln|u|,
\eeq
because (in fully-expanded form),
\begin{align}
\sigma^y_{u}(\mcC;\lambda)&=\lambda (2uu_{x}-v_{x}u_y)+\frac{u_{x}}{\lambda^2 u^2}+\frac{2u_{y}}{u}\ln(\lambda)+\left\{u_{x}v_{y}-2v_{x}v_{yy}+\frac{v_{yy}}{v_{y}}+\frac{2u_{y}}{u}\ln|u|\right\},\label{uxbad}\\
\sigma^y_{v}(\mcC;\lambda)&=\lambda(u_{x}v_{y}-2v_{x}v_{yy}) +\left\{2uu_{x}-v_{x}u_y+\frac{u_{x}}{u^2}+\frac{v_{yy}}{v_{y}}+\frac{2u_{y}}{u}\ln|u|\label{vxbad}\right\}.
\end{align}
The terms in braces are independent of $\lambda$; in \eqref{uxbad} (resp.\ \eqref{vxbad}), some of these depend on $u$ (resp. $v$) and/or its $y$-derivatives. Part of the scaled logarithmic term is independent of $\lambda$, though part survives differentiation. Note that $\sigma^y_{u}$ is a good scaling for the term $u_{x}/u^2$; the singularity at $u=0$ is not an obstacle.
\eex

\begin{lem}\label{goodslem}
The partial scaling $\sigaIx$ is a good scaling for $f(\mbx,[\mbu])$ if and only if
\beq\label{goods2}
f=\lim_{\lambda\rightarrow 1}\int \frac{\upd}{\upd \lambda}\!\left(\sigaIx(f;\lambda)\right)\upd \lambda.
\eeq
\end{lem}

\begin{proof}
	If $\sigaIx$ is a good scaling, \eqref{goods2} is a consequence of $f=\sigaIx(f;1)$ and local smoothness. Conversely, suppose that \eqref{goods2} holds and let $\mu$ be a positive real parameter that is independent of $\lambda$. Then for $\mu$ sufficiently close to $1$,
	\[
		\sigaIx(f;\mu)=\lim_{\lambda\rightarrow 1}\int \frac{\upd}{\upd \lambda}\!\left(\sigaIx(f;\lambda\mu)\right)\upd \lambda=\lim_{\lambda\rightarrow \mu}\int \frac{\upd}{\upd \lambda}\!\left(\sigaIx(f;\lambda)\right)\upd \lambda=\int \frac{\upd}{\upd \mu}\!\left(\sigaIx(f;\mu)\right)\upd \mu.
	\]
	Therefore, $\sigaIx$ is a good scaling.
	
	The use of the limit in \eqref{goods2} is needed to deal with any values of $\uaIk$ for which the integral is an indeterminate form. For other values, simple substitution of $\lambda=1$ gives the limit.
\end{proof}
Given a partial scaling $\sigaIx$ and a differential function $f(\mbx,[\mbu])$, let
\beq\label{proj}
\piaIx(f)=\lim_{\lambda\rightarrow 1}\int \frac{\upd}{\upd \lambda}\!\left(\sigaIx(f;\lambda)\right)\upd \lambda.
\eeq
For $\mu$ sufficiently close to $1$ (using $\widehat{f}$ as shorthand for $\piaIx(f)$),
\[
\sigaIx(\widehat{f};\mu)=\!\int \frac{\upd}{\upd \mu}\!\left(\sigaIx(f;\mu)\right)\!\upd \mu=\!\int\frac{\upd}{\upd \mu}\!\left\{
\int \frac{\upd}{\upd \mu}\!\left(\sigaIx(f;\mu)\right)\upd \mu\right\}\!\upd \mu=\!\int \frac{\upd}{\upd \mu}\!\left(\sigaIx(\widehat{f};\mu)\right)\!\upd \mu;
\] 
the first equality comes from the proof of Lemma \ref{goodslem}. Therefore, $\sigaIx$ is a good scaling for $\piaIx(f)$.
Moreover, there are no terms in the fully-expanded form of the remainder, $f-\piaIx(f)$, for which $\sigaIx$ is a good scaling, because
\[
\frac{\upd}{\upd \mu}\!\left(\sigaIx(f;\mu)\right)-\frac{\upd}{\upd \mu}\!\left(\sigaIx(\widehat{f};\mu)\right)=0.
\]
So $\piaIx$ is the projection that maps a given function onto the component which has $\sigaIx$ as a good scaling.

\begin{defn}
	The partial scaling $\sigaIx$ is a \emph{poor scaling} for a given differential function $f(\mbx,[\mbu])$ if $f-\piaIx(f)$ depends on any $\uaIk\,$.
\end{defn}

For instance, both $\sigma^y_{u}$ and $\sigma^y_v$ are poor scalings for \eqref{awkDiv}, as explained in Example \ref{kerpi}. Section \ref{poors} addresses the problem of inverting divergences such as \eqref{awkDiv} that have poor scalings. First, we develop the inversion process for general divergences. The following results are fundamental.

\begin{thm}\label{fundEul} Let $f(\mbx,[\mbu])$ be a differential function.
	\begin{enumerate}
		\item If $f=D_x F$, then $f\in\mathrm{ker}(\EaIx)$ for all $\alpha$ and $\mbI$; moreover,
		\beq\label{invDF}
		\piaIx(F)=\lim_{\lambda\rightarrow 1}\int \sum_{j\geq 0} \uaIj\,\sigaIx\!\left(\mbE^x_{\ua_{\mbI,j+1}}(f)\,;\lambda\right)\upd \lambda.
		\eeq
		\item If $f\in\mathrm{ker}(\EaIx)$, then $\piaIx(f)\in\mathrm{im}(D_x)$.
		\item If $g=\EaIx(f)$, then, up to terms in $\mathrm{im}(D_x)$,
		\beq\label{invEul}
		\piaIx(f)=\lim_{\lambda\rightarrow 1}\int \uaI\,\sigaIx\!\left(g\,;\lambda\right)\upd \lambda.
		\eeq
	\end{enumerate}
\end{thm}
\begin{proof} All three statements are proved by expanding $\piaIx(f)$:
	\begin{align}
		\piaIx(f)
		&=\lim_{\lambda\rightarrow 1}\int \sum_{j\geq 0}\uaIj\,\sigaIx\!\left(\frac{\p f}{\p \uaIj}\,;\lambda\right)\upd \lambda\label{exp1}\\
		&=\lim_{\lambda\rightarrow 1}\int \uaI\,\sigaIx\!\left(\EaIx(f)\,;\lambda\right)\upd \lambda+D_xh.\label{exp2}
	\end{align}
Here $h(\mbx,[\mbu])$ is obtained by integrating by parts, using the identity \eqref{scalder}.

If $f=D_x F$, the identity \eqref{id1} amounts to $f\in\mathrm{ker}(\EaIx)$. Replace $f$ by $F$ in \eqref{exp1} and use the identity \eqref{id2} to obtain \eqref{invDF}. Statements $2$ and $3$ come directly from \eqref{exp2}.
\end{proof}

Note that \eqref{invDF} is a homotopy formula for (at least partially) inverting $D_xF$, giving a third way to do this. The line integral formula \eqref{line} carries out the full inversion in one step, but may take longer to compute. 

\section{The inversion method for Div}\label{DivInversion}

This section introduces a procedure to invert $\mathrm{Div}$, with a ranking heuristic (informed by experience) that is intended to keep the calculation short and efficient. To motivate the procedure, it is helpful to examine a simple example.

\bex\label{HD}
Wolf \textit{et al.} \cite{WolfBM} introduced computer algebra algorithms that can handle general (non-polynomial) conservation laws, and used these to derive various rational conservation laws of the Harry Dym equation. In the (unique) $x$-dominant positive ranking, the equation is $\mcA=0$, with
\[
\mcA=u_t-u^3u_{,3}\,.
\]
The highest-order conservation law derived in Wolf \textit{et al.} \cite{WolfBM} is $\mcC=\mcQ\mcA$, where
\[
\mcQ=-8uu_{,4}-16u_{,1}u_{,3}-12u_{,2}^2+12u^{-1}u_{,1}^2u_{,2}-3u^{-2}u_{,1}^4\,.
\]
Note that $\sigma^x_{u}$ is a good scaling for $\mcC$. The first step in inverting $\mcC=D_xF+D_tG$ is to apply the partial Euler operator $\Ex$, to annihilate the term $D_xF$. There are only two independent variables, so the identity \eqref{Eulp} shortens the calculation to
\[
\Ex(\mcC)=D_t(\mbE^x_{u_t}(\mcC))=D_t(\mcQ).
\]
Applying \eqref{invEul}, then using Procedure A to integrate by parts (see Example \ref{HDsimp}) gives
\begin{align*}
	\mcQ&=\Ex\!\left\{-\tfrac{8}{3}u^2u_{,4}-\tfrac{16}{3}uu_{,1}u_{,3}-4uu_{,2}^2+4u_{,1}^2u_{,2}-u^{-1}u_{,1}^4\right\}\\
	&=\Ex\!\left\{-4uu_{,2}^2-u^{-1}u_{,1}^4\right\}.
\end{align*}
Therefore
\[
\mcC=D_t\left(-4uu_{,2}^2-u^{-1}u_{,1}^4\right)+\widetilde{\mcC},
\]
where $\Ex(\widetilde{\mcC})=0$. As $\sigma^x_{u}$ is a good scaling for
\[
\widetilde{\mcC}=\mcQ\mcA+8uu_{,2}u_{t,2}+4u^{-1}u_{,1}^3u_{t,1}+\{4u_{,2}^2-u^{-2}u_{,1}^4\}u_t\,,
\]
the second part of Theorem \ref{fundEul} states that $\widetilde{\mcC}\in \mathrm{im}(D_x)$; consequently,
\[
G=-4uu_{,2}^2-u^{-1}u_{,1}^4\,.
\]
Either the line integral \eqref{line} or Procedure \ref{ibp} completes the inversion, giving $\widetilde{\mcC}=D_xF$, where
\[
F=8uu_{,2}u_{t,1}-\{8uu_{,3}+8u_{,1}u_{,2}-4u^{-1}u_{,1}^3\}u_t+4u^4u_{,3}^2+4u^3u_{,2}^3-6u^2u_{,1}^2u_{,2}^2+3uu_{,1}^4u_{,2}-\tfrac{1}{2}u_{,1}^6.
\]
The fully-expanded form of $(F,G)$ is minimal, having $11$ terms rather than the $12$ terms in Wolf \textit{et al.} \cite{WolfBM}. (Note: there is an equivalent conservation law, not in the form $\mcQ\mcA$, that has only $10$ terms.)
\eex

\subsection{A single iteration}\label{iteration}

The basic method for inverting a given total divergence one independent variable at a time works similarly to the example above. Suppose that after $n$ iterations the inversion process has yielded components $F^i_n$ and that an expression of the form $\mcC=D_if_n^i$ remains to be inverted. For the next iteration, let $\Ex$ be the partial Euler operator that is applied to $\mcC$. Here $u$ is one of the variables $\uaI$, which is chosen to ensure that for each $i$ such that $x^i\neq x$,
\beq\label{EDcom}
\Ex D_if_n^i =D_i\Ex f_n^i.
\eeq
This requires care, in view of the identity \eqref{id3}. However, it is achievable by using the variables $\uaI$ in the order given by a ranking that is discussed in Section \ref{rankvars}. This ranking is entirely determined by user-defined rankings of the variables $x^j$ and $u^\alpha$.

Taking \eqref{id1} into account leads to the identity
\beq\label{Eulp0}
\Ex(\mcC)=\sum_{x^i\neq x}D_i(\Ex(f_n^i)),
\eeq
which, with together with Theorem \ref{fundEul}, is the basis of the inversion method. The method works without modification provided that:
\begin{itemize}
	\item there are no poor scalings for any terms in $\mcC$;
	\item the fully-expanded form of $\mcC$ has no terms that are linear in $[\mbu]$.
\end{itemize}
We begin by restricting attention to divergences for which these conditions hold, so that
\beq\label{Eulp1}
\Ex(\mcC)=\sum_{x^i\neq x}D_i(\Ex(\pix f_n^i)),
\eeq
where every term in $\Ex(\pix f_n^i)$ depends on $[\mbu]$. The modifications needed if either condition does not hold are given in Sections \ref{poors} and \ref{lindiv}.

The iteration of the inversion process runs as follows. Calculate $\Ex(\mcC)$, which is a divergence $D_iP^i$ with no $D_x$ term, by \eqref{Eulp1}; it involves at most $p-1$ (but commonly, very few) nonzero functions $P^i$. Invert this divergence, treating $x$ as a parameter. \textit{If it is possible to invert in more than one way, always invert into the $P^i$ for which $x^i$ is ranked as low as possible}; the reason for this is given in the next paragraph.  If $\Ex(\mcC)$ has nonlinear terms that involve derivatives with respect to more than one $D_i$ (excluding $D_x$), this is accomplished by iterating the inversion process with as few independent variables as are needed. Otherwise, $P^i$ can be determined more quickly by using integration by parts (Procedure \ref{ibp}, with $x$ replaced by the appropriate $x^i$), and/or the method for linear terms (see Procedure \ref{lininv} in Section \ref{lindiv}). Note that this shortcut can be used whenever there are only two independent variables.

\textit{At this stage, check that the fully-expanded form of each $P^i$ has no terms that are ranked lower than $u$.} If any term is ranked lower than $u$, stop the calculation and try a different ranking of the variables $x^j$ and/or $u^\alpha$. This is essential because, to satisfy \eqref{EDcom} and avoid infinite loops, the variables $\uaI$ that are chosen to be $u$ in successive iterations must progress upwards through the ranking. Where there is a choice of inversion, the rank of each term in $P^i$ is maximized by using the $x^i$ of minimum order; this avoids unnecessary re-ranking.  

Having found and checked $P^i$, use \eqref{exp2} to obtain
\beq\label{invEul1}
\pix(f_n^i)=\left\{\lim_{\lambda\rightarrow 1}\int u\,\sigx\!\left(P^i\,;\lambda\right)\upd \lambda\right\} +D_xh^i,
\eeq
for arbitrary functions $h^i(\mbx,[\mbu])$. Apply Procedure \ref{ibp} to the function in braces and choose $h^i$ to make the right-hand side of \eqref{invEul1} equal the remainder from this procedure. This yields the representation of $\pix(f_n^i)$ that has the lowest-order derivatives (with respect to $x$) consistent with the inversion of $D_iP^i$; call this representation $f^i$. Commonly, such a lowest-order representation is needed to obtain a minimal inversion.

By Theorem \ref{fundEul}, there exists $\phi$ such that
\beq\label{Fdef}
\pix\left(\mcC-\sum_{x^i\neq x}D_if^i\right)=D_x\phi,
\eeq
because (by construction) the expression in parentheses belongs to $\mathrm{ker}(\Ex)$. 
Use the line integral formula \eqref{line} or Procedure \ref{ibp} to obtain $\phi$, then set $f^i:=\phi$ for $x^i=x$. Now update: set
\[
\mcC:=\mcC-D_if^i,\qquad F^i_{n+1}:=F^i_n+f^i.
\]

\subsection{Ranking and using the variables}\label{rankvars}

Having described a single iteration, we now turn to the question of how to choose $x$ and $u$ effectively. The starting-point is to construct a \textit{derivative-dominant} ranking of the variables $\uaJ$. This is a positive ranking that is determined by:
\begin{itemize}
	\item a ranking of the independent variables, $x^1\prec x^2\prec \cdots \prec x^p$;
	\item a ranking of the dependent variables, $u^1\prec u^2\prec \cdots \prec u^q$.
\end{itemize}
(Later in this section, we give a heuristic for ranking the dependent and independent variables effectively.)
The derivative-dominant ranking (denoted $\mbu_p$) is constructed iteratively, as follows.
\begin{align*}
\mbu_0&=u^1\prec\cdots\prec u^q,\\
\mbu_1&=\mbu_0\prec D_1\mbu_0 \prec D_1^2\mbu_0 \prec\cdots ,\\
\mbu_2&=\mbu_1\prec D_2\mbu_1 \prec D_2^2\mbu_1 \prec\cdots ,\\
\vdots&\qquad\vdots\\
\mbu_p&=\mbu_{p-1}\prec D_p\mbu_{p-1} \prec D_p^2\mbu_{p-1} \prec\cdots.
\end{align*}
In practice, very few $\uaJ$ are needed to carry out many inversions of interest, but it is essential that these are used in the order given by their ranking, subject to a constraint on $|\mbI|$ that is explained below. 

Given an independent variable, $x$, we call $\uaI$ \textit{relevant} if the updated $\mcC$ depends on $\uaIk$ for some $k\geq 0$. The first set of iterations uses $x=x^1$. For the initial iteration, $u$ is the lowest-ranked relevant $\ua$. In the following iteration, $u$ is the next-lowest-ranked relevant $u^\alpha$ and so on, up to and including $u^q$. (From \eqref{id3}, the condition \eqref{EDcom} holds whenever $u=\ua,\ \alpha = 1,\dots, q$.) After these iterations, the updated $\mcC$ is independent of $\mbu$ and its unmixed $x$-derivatives.

If the updated $\mcC$ has any remaining $x$-derivatives, these are mixed. Thus, as $\mcC$ has no linear terms, a necessary condition for the inversion to be minimal is that every $f_n^i$ is independent of $\mbu$ and its unmixed $x$-derivatives. Consequently, \eqref{EDcom} holds for $u=\uaI$ whenever $|\mbI|=1$, because
\[
\mbE^x_{\ua_{\mbI-\mbone_i}}(f_n^i)=0,\qquad x^i\neq x.
\]
Therefore, the process can be continued using each relevant $u=\uaI$ with $|\mbI|=1$ in the ranked order. Iterating, the same argument is used with $|\mbI|=2,3,\dots$, until $\mcC$ is independent of $x$-derivatives. Now set $x=x^2$ and iterate, treating $x^1$ as a parameter. In principle, this can be continued up to $x=x^p$; in practice, only a very few iterations are usually needed to complete the inversion. The best rankings invert many terms at each iteration. On the basis of some experience with conservation laws, the following heuristic for ranking the variables  $x^j$ and $\ua$ is recommended.
 
%The components of the divergence on the right-hand side of \eqref{Eulp1} come from
%the identity \eqref{Eulp}, which gives
%\beq\label{Eulp2}
%\Eax(\mcC)=-\sum_{I\neq 0}(-D)_I\EaIx(\mcC).
%\eeq

\begin{defn2} Apply criteria for ranking independent variables, using the following order of precedence.
	\begin{enumerate}
		\item Any independent variables that occur in the arguments of arbitrary functions of $\mbx$ should be ranked as high as possible, if they multiply terms that are nonlinear in $[\mbu]$. For instance, if nonlinear terms in a divergence depend on an arbitrary function, $f(t)$, set $x^p=t$.
		\item If independent variables occur explicitly in non-arbitrary functions, they should be ranked as high as possible (subject to 1 above), with priority going to variables with the most complicated functional dependence. For instance, if $\mcC$ is linear in $x^i$ and quadratic in $x^j$, then $x^i\prec x^j$ (so $i<j$ in our ordering).
		\item If an unmixed derivative of any $u^\alpha$ with respect to $x^i$ is the argument of a function other than a rational polynomial, rank $x^i$ as low as possible.
		\item Set $x^i\prec x^j$ if the highest-order \textit{unmixed} derivative (of any $u^\alpha$) with respect to $x^i$ is of higher order than the highest-order unmixed derivative with respect to $x^j$.
		\item Set $x^i\prec x^j$ if there are more occurrences of unmixed $x^i$-derivatives (in the fully-expanded divergence) than there are of unmixed $x^j$-derivatives.
		\item Apply criteria $3$, $4$, and $5$ in order of precedence, replacing unmixed by `minimally-mixed' derivatives. Minimally-mixed means that there are as few derivatives as possible with respect to any other variable(s).
	\end{enumerate}
The derivative indices $\mbJ$ in a derivative-dominant ranking are ordered according to the ranking of the independent variables. This can be used to rank the dependent variables; if there is more than one dependent variable in $\mcC$, use the following criteria in order.
	\begin{enumerate}
	\item Let $u^\alpha\prec u^\beta$ if $\mcC$ is linear in $[\ua]$ and nonlinear in $[u^\beta]$. 
	\item Let $u^\alpha\prec u^\beta$ if the lowest-ranked derivative of $u^\alpha$ that occurs in $\mcC$ is ranked lower than the lowest-ranked derivative of $u^\beta$ in $\mcC$. (In conservation laws, the lowest-ranked derivative of $u^\alpha$ is commonly the undifferentiated $\ua$, which corresponds to $\mbJ=0$.)
	\item Let $u^\alpha\prec u^\beta$ if the lowest-ranked derivative of $u^\alpha$ in the fully-expanded form of $\mcC$ occurs in more terms than the corresponding derivative of $u^\beta$ does.
	\end{enumerate}
These two sets of criteria are not exhaustive (allowing ties, which must be broken), but the aim that underlies them is to carry as few terms as possible into successive iterations of the procedure. Partial Euler operators with respect to unmixed derivatives are used in the earliest iterations; commonly, these are sufficient to complete the inversion.
\end{defn2}

\subsection{How to deal with poor scalings}\label{poors}

To remove an earlier restriction on the inversion process, we now address the problem of poor scalings, namely, that $\uaI$ and its $x$-derivatives (denoted $[\uaI]_x$) may occur in terms that belong to $\mcC-\piaIx(\mcC)$. Such terms (when fully expanded) are products of homogeneous rational polynomials in $[\uaI]_x$ of degree zero and logarithms of a single element of $[\uaI]_x$. We refer to these collectively as \textit{zero-degree} terms.

To overcome this difficulty, we modify the approach used by Anco \& Bluman \cite{AncBlu2} to treat singularities. In our context, $[\uaI]_x$ is replaced by $[\uaI+U^\alpha_\mbI]_x$, where $U^\alpha_\mbI$ is regarded as a new dependent variable that is ranked higher than $\uaI$. This approach works equally well for logarithms, ensuring that $\piaIx$ is a good scaling for all terms that depend on $[\uaI]_x$, so that its kernel consists only of terms that are independent of these variables. At the end of the calculation, all members of $[U^\alpha_\mbI]_x$ are set to zero. Note that there is no need to replace $\uaI$ in terms that are not zero-degree in $[\uaI]_x$, as total differentiation preserves the degree of homogeneity\footnote{An alternative approach (see Hickman \cite{Hickman}) uses a locally-invertible change of variables, $u^\alpha= \exp{v^\alpha}$, to change the degree of homogeneity from zero. This approach works equally well, but requires a little more processing time.}.

\bex
To illustrate the inversion process for divergences that have zero-degree terms, we complete Example \ref{kerpi} by inverting
\[
\mcC=D_xF+D_yG=u_{x}(2u+v_{y})-v_{x}(u_y+2v_{yy})+\frac{u_{x}}{u^2}+\frac{v_{yy}}{v_{y}}+\frac{2u_{y}}{u}\ln|u|;
\]
this has no linear terms. The ranking heuristic gives $y\prec x$ and $u\prec v$. The rest of the calculation goes as follows.

\begin{enumerate}
	\item $\mcC-\pi_u^y(\mcC)$ has just one zero-degree term (in $[u]_y$), namely $(2u_y/u)\ln|u|$. Replace this term by $(2(u_y+U_y)/(u+U))\ln|u+U|$.
	\item Calculate $\mbE_u^y(\mcC)=2u_x+v_{xy}-2u_xu^{-3}=D_x\{2u+v_y+u^{-2}\}=D_x\{\mbE_u^y(u^2+uv_y-u^{-1})\}$. No term in $2u+v_y+u^{-2}$ is ranked lower than $u$, as $u$ is the lowest-ranked variable.
	\item Then
	$\pi_u^y(\mcC-D_x\{u^2+uv_y-u^{-1}\})=D_y\{-uv_x+(\ln|u+U|)^2\}$.
	\item Now set $U=0$ to yield the remainder $\mcC_1=\mcC-D_x\{u^2+uv_y-u^{-1}\}-D_y\{-uv_x+(\ln|u|)^2\}$ at the close of the first iteration. This amounts to $\mcC_1=-2v_xv_{yy}+v_{yy}/v_y$.
	\item The second iteration starts with $\mcC_1-\pi_v^y(\mcC_1)=v_{yy}/v_y$; as this term is zero-degree in $[v]_y$, replace it by $(v_{yy}+V_{yy})/(v_y+V_y)$.
	\item Calculate $\mbE_v^y(\mcC_1)=-2v_{xyy}=D_x\{-2v_{yy}\}=D_x\{E_v^y(-vv_{yy})\}=D_x\{E_v^y(v_y^2)\}$. Note that $-2v_{yy}$ is not ranked lower than $v$.
	\item Then
	$\pi_v^y(\mcC_1-D_x\{v_y^2\})=D_y\{-2v_xv_y+\ln|v_y+V_y|\}$.
	\item Now set $V=0$ to yield the remainder $\mcC_2=\mcC_1-D_x\{v_y^2\}-D_y\{-2v_xv_y+\ln|v_y|\}=0$ at the close of the second iteration. The inversion process stops, having yielded the output
	\[
	F=u^2+uv_y-u^{-1}+v_y^2,\qquad G=-uv_x+(\ln|u|)^2-2v_xv_y+\ln|v_y|.
	\]
\end{enumerate}

Note that the homotopy formula \eqref{invDF} for inverting $D_xF$ can be adjusted in the same way, whenever $\sigaIx$ is a poor scaling for $\mbE^x_{\ua_{\mbI,j+1}}(D_xF)$.
\eex

\subsection{Linear divergences}\label{lindiv}

The inversion process runs into a difficulty when a given divergence has terms that are linear in $[\mbu]$, with mixed derivatives. Then it is possible to invert in more than one way, some of which may not produce a minimal result. To address this, it is helpful to invert using a different process. Suppose that $\mcC$ is a linear divergence (in fully-expanded form). Instead of using the derivative-dominant ranking, integrate by parts, working down the total order $|\mbJ|$ of the derivatives $\uaJ$. For a given total order, we will invert the mixed derivatives first, though this is not essential. 

Starting with the highest-order derivatives, one could integrate each term $f(\mbx)\uaJ$ in $\mcC$ by parts with respect to any $x^i$ such that $j^i\geq 1$, yielding the remainder $-D_i(f)\ua_{\mbJ-\mbone_i}$. If $D_\mbJ$ is a mixed derivative, we seek to choose $D_i$ in a way that keeps the result concise. Here are some simple criteria that are commonly effective, listed in order of precedence.
\begin{enumerate}
	\item $f$ is independent of $x^i$.
	\item $\mcC$ includes the term $D_i(f)\ua_{\mbJ-\mbone_i}$.
	\item $f$ is linear in $x^i$.
\end{enumerate}
These criteria can be used as an initial pass to invert $\mcC$ at least partially, leaving a remainder to be inverted that may have far fewer terms than $\mcC$ does.

Integrating the remainder by parts is straightforward if each $\uaJ$ is an unmixed derivative. If $\mbJ$ denotes a mixed derivative, integrate with respect to each $x^i$ such that $j^i\geq 1$ in turn, multiplying each result by a parameter (with the parameters summing to $1$). Iterate until either the remainder has a factor that is zero for some choice of parameters or there is no remainder. The final stage is to choose the parameters so as to minimize the number of terms in the final expression. Although this produces a minimal inversion, it comes at the cost of extra computational time spent doing all possible inversions followed by the parameter optimization.

\bex\label{linin}
Consider the linear divergence
\[
\mcC=\left(\tfrac{1}{6}f'(t)y^3+f(t)xy\right)\!(u_{xt}-u_{yy}),
\]
where $f$ is an arbitrary function. The first of the simple criteria above yields
\[
\mcC=D_x\!\left\{\tfrac{1}{6}f'y^3u_t\right\}+fxyu_{xt}-\left(\tfrac{1}{6}f'y^3+fxy\right)\!u_{yy};
\]
The second criterion is not helpful at this stage, but the third criterion gives
\[
\mcC=D_x\!\left\{\left(\tfrac{1}{6}f'y^3+fxy\right)\!u_t\right\}-fyu_{t}-\left(\tfrac{1}{6}f'y^3+fxy\right)\!u_{yy}.
\]
The remainder has no mixed derivatives; integrating it by parts produces the minimal inversion
\beq\label{KZlin}
\mcC=D_x\!\left\{\left(\tfrac{1}{6}f'y^3+fxy\right)\!u_t\right\}+D_y\!\left\{\left(\tfrac{1}{2}f'y^2+fx\right)\!u-\left(\tfrac{1}{6}f'y^3+fxy\right)\!u_y\right\}+D_t\!\left\{-fyu\right\}.
\eeq
%\begin{align*}
%\mcC&=D_x\!\left\{\left(\tfrac{1}{6}f'y^3+fxy\right)\!u_t\right\}+D_y\!\left\{-\left(\tfrac{1}{6}f'y^3+fxy\right)\!u_y\right\}-fyu_{t}+\left(\tfrac{1}{2}f'y^2+fx\right)\!u_{y}\\
%&=D_x\!\left\{\left(\tfrac{1}{6}f'y^3+fxy\right)\!u_t\right\}+D_y\!\left\{\left(\tfrac{1}{2}f'y^2+fx\right)\!u-\left(\tfrac{1}{6}f'y^3+fxy\right)\!u_y\right\}-fyu_{t}-f'yu\\
%&=D_x\!\left\{\left(\tfrac{1}{6}f'y^3+fxy\right)\!u_t\right\}+D_y\!\left\{\left(\tfrac{1}{2}f'y^2+fx\right)\!u-\left(\tfrac{1}{6}f'y^3+fxy\right)\!u_y\right\}+D_t\!\left\{-fyu\right\}.
%\end{align*}

\eex

\bex\label{linin2}
To illustrate the parametric approach, consider
\[
\mcC=\exp(t-x^2)(tu_{xtt}+2x(t+1)u_t).
\]
The simple criteria are irrelevant at present, so instead introduce parameters $\lambda_l$ and consider all possible inversions of the mixed derivative terms. Step-by-step, one obtains the following. 
\begin{align*}
\mcC&=D_x\!\left\{\lambda_1\exp(t-x^2)tu_{tt}\right\}+D_t\!\left\{(1-\lambda_1)\exp(t-x^2)tu_{xt}\right\}\\
&\quad+\exp(t-x^2)\{2\lambda_1xtu_{tt}-(1-\lambda_1)(t+1)u_{xt}+2x(t+1)u_t\}\\
&=D_x\!\left\{\exp(t-x^2)\{\lambda_1tu_{tt}-\lambda_2(1-\lambda_1)(t+1)u_{t}\}\right\}\\
&\quad+D_t\!\left\{\exp(t-x^2)\{(1-\lambda_1)tu_{xt}+2\lambda_1xtu_{t}-(1-\lambda_2)(1-\lambda_1)(t+1)u_{x}\}\right\}\\
&\quad+(1-\lambda_2)(1-\lambda_1)\exp(t-x^2)\{2x(t+1)u_{t}+(t+2)u_{x}\}.
\end{align*}
As the remainder has a factor $(1-\lambda_2)(1-\lambda_1)$, the inversion is complete if either parameter is set to $1$. The minimal (two-term) inversion has $\lambda_1=1$, which gives $\mcC=D_xF+D_tG$, where 
\[
F=t\exp(t-x^2)u_{tt}\,,\qquad G=2xt\exp(t-x^2)u_{t}\,.
\]
For $\lambda_1\neq 1, \lambda_2=1$, the inversion has three terms if $\lambda_1=0$, or five terms otherwise.

Note the importance of factorizing the remainder to stop the calculation once a possible parameter choice occurs. If we had continued the calculation without setting either $\lambda_i$ to $1$, it would have stopped at order zero, not one, giving an eleven-term inversion for general $\lambda_i$.
\eex

To summarize, one can invert a divergence that is linear in $[\mbu]$ by applying the following procedure.

\begin{rout}\label{lininv} \textbf{Inversion of a linear total divergence}  
		\begin{enumerate}
			\item[]
			\begin{enumerate}
				\item[\textsc{Step} 0.] Identify the maximum derivative order, $N=|\mbJ|$, of the variables $\uaJ$ that occur in $\mcC$. Set $F^i:=0,\ i=1\dots,p$.
				\item[\textsc{Step} 1.] While $\mcC$ has at least one term of order $N$, do the following. Select any such term, $f(\mbx)\uaJ$, and determine a variable $x^i$ over which to integrate. If desired, parametrize for mixed derivatives, as described above. Set $F^i:=F^i+f\ua_{\mbJ-\mbone_i}$ for the chosen $i$ (with the appropriate modification if parameters are used), and update the remainder by setting $\mcC:=\mcC-f(\mbx)\uaJ-D_i(f)\ua_{\mbJ-\mbone_i}$. Once $\mcC$ has no terms of order $N$, continue to \textsc{Step} 2.
				\item[\textsc{Step} 2.] If $\mcC$ is nonzero and cannot be set to zero by any choice of parameters, set $N:=N-1$ and return to \textsc{Step} 1; otherwise, set $\mcC$ to zero and carry out parameter optimization (if needed), give the output $F^i,\ i=1,\dots,p$, then \textbf{stop}.
			\end{enumerate}
		\end{enumerate}
\end{rout}

\subsection{Summary: a procedure for inverting Div}

Having addressed potential modifications, we are now in a position to summarize the inversion process for any total divergence $\mcC$ whose fully-expanded form has no terms depending on $\mbx$ only.

\begin{rout}\label{inversion} \textbf{Inversion of $\mcC=D_iF^i$} 
	\begin{enumerate}
		\item[] 
		\begin{enumerate}
			\item[]
			\begin{enumerate}
				\item[\textsc{Step} 0.] Let $\mcC_\ell$ be the linear part of $\mcC$. If $\mcC_\ell=0$, set $\mcC_0:=\mcC$ and $F_0^i:=0,\ i=1,\dots,p$. Otherwise, invert $\mcC_\ell$ using the technique described in Procedure \ref{lininv} above, to obtain
				functions $F_0^i$ that satisfy $\mcC_\ell=D_iF_0^i$. Set $\mcC_0:=\mcC-\mcC_\ell$. Choose a derivative-dominant ranking for $\mcC_0$ (either using the ranking heuristic or otherwise). In the notation used earlier, set $x:=x^1$ and $u$ to be the lowest-ranked relevant $\ua$; typically, $u:=u^1$. Set $n:=0$; here $n+1$ is the iteration number.
				\item[\textsc{Step} 1.] Calculate $\mcC_n-\pix(\mcC_n)$; if this includes terms depending on $[u]_x$, replace $[u]_x$ \textit{in these terms only} by $[u+U]_x$.
				\item[\textsc{Step} 2.] Apply the process detailed in Section \ref{iteration} (with $\mcC_n$ replacing $\mcC$). Provided that the ranking check is passed, this yields components $f^i$, which may depend on $[U]_x$. If the ranking check is failed, choose a different ranking of $x^j$ and $\ua$ for the remainder of the inversion process and return to \textsc{Step} 1, starting with the lowest-ranked $x$ and (relevant) $u$ and working upwards at each iteration. 
				\item[\textsc{Step} 3.] Replace all elements of $[U]_x$ in $f^i$ by zero.
				\item[\textsc{Step} 4.] Update: set $F^i_{n+1}:=F^i_n+f^i,\ \mcC_{n+1}:=\mcC_n-D_if^i$ and $n:=n+1$.
				If $\mcC_{n+1}= 0$, output $F^i=F^i_{n+1}$ and \textbf{stop}. Otherwise, update $u$ and $x$ as detailed in Section \ref{rankvars} and return to \textsc{Step} 1.
			\end{enumerate}
		\end{enumerate}
	\end{enumerate}
\end{rout}

\bex
As an example with linear terms, the Khokhlov--Zabolotskaya (KZ) equation,
\[
u_{xt}-uu_{xx}-u_x^2-u_{yy}=0,
\]
has a conservation law (see Poole \& Hereman \cite{Poole11}) that involves an arbitrary function, $f(t)$:
\[
\mcC=D_xF+D_yG+D_tH=\left(\tfrac{1}{6}f'y^3+fxy\right)\!\left(u_{xt}-uu_{xx}-u_x^2-u_{yy}\right).
\]
The linear part $\mcC_\ell$ of  $\mcC$ is inverted by Procedure \ref{lininv}, as shown in Example \ref{linin}. From \eqref{KZlin}, \textsc{Step} 0 in Procedure \ref{inversion} gives
\[
F_0=\left(\tfrac{1}{6}f'y^3+fxy\right)\!u_t\,,\qquad G_0=\left(\tfrac{1}{2}f'y^2+fx\right)\!u-\left(\tfrac{1}{6}f'y^3+fxy\right)\!u_y\,, \qquad H_0=-fyu.
\]
The remainder $\mcC-\mcC_\ell$ is
\[
\mcC_0=-\left(\tfrac{1}{6}f'y^3+fxy\right)\!\left(uu_{xx}+u_x^2\right).
\]
With the derivative-dominant ranking $x\prec y\prec t$, only one iteration is needed to complete the inversion, because $\Ex(\mcC_0)=0$. Then 
\[
\pi_u^x(\mcC_0)=D_x\!\left\{\tfrac{1}{2}fyu^2-\left(\tfrac{1}{6}f'y^3+fxy\right)\!uu_x\right\},
\]
and as $\mcC_0=\pi_u^x(\mcC_0)$, the calculation stops after updating, giving the output
\[
F=\left(\tfrac{1}{6}f'y^3\!+\!fxy\right)\!\left(u_t\!-\!uu_x\right)+\tfrac{1}{2}fyu^2,\qquad G=\left(\tfrac{1}{2}f'y^2\!+\!fx\right)\!u-\left(\tfrac{1}{6}f'y^3\!+\!fxy\right)\!u_y,\qquad H=-fyu.
\]
This corresponds to the minimal result in Poole \& Hereman \cite{Poole11}.
\eex

\bex
In every example so far, all iterations have used $u=\ua$, but not $\uaI$ for $|I|\neq 0$. The BBM equation from Example \ref{BBM} illustrates the way that the process continues unhindered once the current $\mcC$ is independent of all elements of $[\ua]_x$. The conservation law \eqref{BBMeq} is
\[
\mcC=D_xF+D_tG=(u^2\!+\!2u_{xt})(u_{t}\!-\!uu_x\!-\!u_{xxt});
\]
it has no linear part or zero-degree terms. The ranking heuristic gives $x\prec t$. The first iteration using $\Ex$ is similar to what we have seen so far, giving the following updates:
\[
\mcC_1=2u_{xt}(u_t-u_{xxt}),\qquad F_1=-\tfrac{1}{4}u^4-u^2u_{xt},\qquad G_1=\tfrac{1}{3}u^3.
\]
At this stage, $\mcC_1$ is independent of $[u]_x$ and equals $\pi^x_{u_t}(\mcC_1)$. In the second iteration, the ranking requires us to apply $\mbE^x_{u_t}$. This annihilates $\mcC_1$; consequently, $\mcC_1$ is a total derivative with respect to $x$. Inverting this gives the final (minimal) result,
\[
F=-\tfrac{1}{4}u^4-u^2u_{xt}-u_{xt}^2+u_t^2,\qquad G=\tfrac{1}{3}u^3,
\]
which was obtained by inspection in Example \ref{BBM}.
\eex

The Appendix lists inversions of various other divergences; in each case, the inversion produced by Procedure \ref{inversion} and the ranking heuristic is minimal and takes very few iterations to complete.

\subsection{Splitting a divergence using discrete symmetries}

A given divergence may have discrete symmetries between various terms in its fully-expanded form. If the divergence has very many terms that are connected by a particular discrete symmetry group, it can be worth splitting these into disjoint divergences that are mapped to one another by the group elements. Then it is only necessary to invert one of these divergences, using the symmetries to create inversions of the others without the need for much extra computation. However, to use Procedure \ref{inversion}, it is necessary to check that all split terms are grouped into divergences; this check is done by using the criterion \eqref{divcrit}.

Polynomial divergences can first be split by degree, yielding divergences that are homogeneous in $[\mbu]$. Such splitting does not add significantly to the computation time, nor does it need to be checked using \eqref{divcrit}, which holds automatically. Splitting by degree can make it easy to identify terms that are linked by discrete symmetries, as illustrated by the following example.

\bex

In Cartesian coordinates $(x,y)$, the steady non-dimensionalized von K\'arm\'an equations for a plate subject to a prescribed axisymmetric load function, $p(x^2+y^2)$, are $\mcA_\ell=0$, $\ell=1,2$, where
\begin{align*}
	\mcA_1&=u_{xxxx}+2u_{xxyy}+u_{yyyy}-u_{xx}v_{yy}+2u_{xy}v_{xy}-u_{yy}v_{xx}-p,\\
		\mcA_2&=v_{xxxx}+2v_{xxyy}+v_{yyyy}+u_{xx}u_{yy}-u_{xy}^2\,.
\end{align*} 
Here $u$ is the displacement and $v$ is the Airy stress. This is a system of Euler--Lagrange equations. By Noether's Theorem, the one-parameter Lie group of rotational symmetries yields the following conservation law (see Djondjorov and Vassilev \cite{DV00}):
\[
\mcC=(yv_x-xv_y)\mcA_2-(yu_x-xu_y)\mcA_1\,.
\]
This conservation law has linear, quadratic and cubic terms, so $\mcC$ can be inverted by summing the inversions of each of the following divergences:
\begin{align*}
	\mcC_\ell&=(yu_x-xu_y)p,\\
	\mcC_q&=(yv_x-xv_y)(
	v_{xxxx}+2v_{xxyy}+v_{yyyy})-(yu_x-xu_y)(u_{xxxx}+2u_{xxyy}+u_{yyyy}),\\
	\mcC_c&=(yv_x-xv_y)(u_{xx}u_{yy}-u_{xy}^2)+(yu_x-xu_y)(u_{xx}v_{yy}-2u_{xy}v_{xy}+u_{yy}v_{xx}).
\end{align*}
The quadratic terms have an obvious discrete symmetry, $\Gamma_1:(x,y,u,v)\mapsto(-x,y,v,u)$, which gives a splitting into two parts, each of which is a divergence:
\[
\mcC_q=\overline{\mcC}_q+\Gamma_1(\overline{\mcC}_q),
\]
where
\[
\overline{\mcC}_q=(yv_x-xv_y)(v_{xxxx}+2v_{xxyy}+v_{yyyy}).
\]
Consequently, we can invert $\mcC_q$ by inverting $\overline{\mcC}_q$ and applying the symmetry $\Gamma_1$.

Note that $\overline{\mcC}_q$ has a discrete symmetry, $\Gamma_2:(x,y,v)\mapsto(y,-x,v)$, which gives
\[
\overline{\mcC}_q=g+\Gamma_2(g),\quad\text{where}\quad g=(yv_x-xv_y)(v_{xxxx}+v_{xxyy}).
\]
Checking \eqref{divcrit} shows that this is not a valid splitting into divergences, because
\[
\mbE_v(g)=-2v_{xxxy}-2v_{xyyy}\neq 0.
\]

If necessary, split divergences can be inverted using different rankings. However, in this simple example, a single ranking works for all (nonlinear) parts. One tie-break is needed: let $x\prec y$. The ranking heuristic gives $v\prec u$ for the cubic terms; the variables $[u]$ are not relevant in the inversion of $\overline{\mcC}_q$. Procedure \ref{inversion} gives the following minimal inversions,
\begin{align*}
\mcC_\ell&=D_x(yup)+D_y(-xup),\\
\overline{\mcC}_q&=D_x\{F_q(x,y,[v])\}+D_y\{G_q(x,y,[v])\},\\
\mcC_c&=D_xF_c+D_yG_c\,,
\end{align*}
where
\begin{align*}
	F_q(x,y,[v])=&\ y\left\{v_xv_{xxx}-\tfrac{1}{2}v_{xx}^2+2v_xv_{xyy}+v_{xy}^2+vv_{yyyy}-\tfrac{1}{2}v_{yy}^2\right\}\\
	&+x\left\{-v_yv_{xxx}+v_{xx}v_{xy}-2v_yv_{xyy}\right\}+v_yv_{xx}\,,\\
	G_q(x,y,[v])=&\ y\left\{-2v_{xx}v_{xy}-vv_{xyyy}+v_yv_{xyy}\right\}+x\left\{-\tfrac{1}{2}v_{xx}^2+v_{xy}^2-v_yv_{yyy}+\tfrac{1}{2}v_{yy}^2\right\}+vv_{xyy}\,,\\
	F_c=&\
	y\left(u_xu_{yy}v_x+\tfrac{1}{2}u_x^2v_{yy}\right)+x\left(u_yu_{xy}v_y+\tfrac{1}{2}u_y^2v_{xy}\right)-\tfrac{1}{2}u_y^2v_y\,,\\
	G_c=&\
	-y\left(u_xu_{xy}v_x+\tfrac{1}{2}u_x^2v_{xy}\right)-x\left(u_yu_{xx}v_y+\tfrac{1}{2}u_y^2v_{xx}\right)+\tfrac{1}{2}u_x^2v_x\,.
\end{align*}
Applying $\Gamma_1$ to $\overline{\mcC}_q$ gives the following minimal inversion for $\mcC$:
\[
\mcC=D_x\!\left\{yup +F_q(x,y,[v])-F_q(x,y,[u]) +F_c\right\}+D_y\!\left\{-xup +G_q(x,y,[v])-G_q(x,y,[u]) +G_c\right\}.
\]

Note that the tie-break $x\prec y$ causes the inversion to break the symmetry $\Gamma_2$ that is apparent in $\overline{\mcC}_q$. As it turns out, this symmetry can be restored without increasing the overall number of terms, by adding components of a trivial conservation law (which are of the form \eqref{trans}). It is an open question whether such preservation of symmetry and minimality is achievable in general.

In Djondjorov \& Vassilev \cite{DV00}, the inversion of $\mcC$ for the special case $p=0$ has 62 terms, which are grouped according to their physical meaning. By contrast, the minimal inversion above has just 46 (resp. 48) terms when $p$ is zero (resp. nonzero), a considerable saving. Moreover, by exploiting the symmetry $\Gamma_1$, only 28 (resp. 30) of these terms are determined using Procedure \ref{inversion}. However, in seeking an efficient inversion, we have ignored the physics. It would be interesting to understand the extent to which the use of a minimal inversion obscures the underlying physics.

\eex

\section{Concluding remarks}

Partial Euler operators and partial scalings make it possible to invert divergences with respect to one independent variable at a time, a by-product being that some contributions to other components are determined at each iteration step. Although each iteration involves a fair amount of computation, very few iterations are needed for many systems of interest.

Given the potential complexity of functions, it is unlikely that every divergence can be inverted, even in principle, by Procedure \ref{inversion}. The question of how to prove or disprove this is open. In practice, products of mixed derivatives present the greatest challenge to concise inversion, although the option of re-ranking part-way through the procedure enables a divide-and-conquer approach to be taken.

The focus of this work has been on inverting the total divergence operator Div. However, this immediately applies to expressions that can be recast as total divergences. For instance, for $p=3$, the total curl $\mbF=\mathrm{Curl}(\mathbf{G})$ can be inverted by writing
\beq\label{Curlinv}
F^i=\mathrm{Div}(H^{ij}\mathbf{e}_j)=D_jH^{ij},\qquad H^{ij}:=\epsilon^{ijk}G_k=-H^{ji},
\eeq
where $\epsilon^{ijk}$ is the Levi--Civita symbol, then inverting one component at a time and using the results at each stage to simplify the remainder of the calculation. Once $H^{ij}$ is known, the identity $G_l=\frac{1}{2}\epsilon_{ijl}H^{ij}$ recovers $\mathbf{G}$. Typically, a minimal inversion is achieved by using a different ranking for each $F^i$, in accordance with the ranking heuristic. Here is a simple illustration of the general approach.

\bex
In Cartesian coordinates, invert $(F^x,F^y,F^z)=\mathrm{Curl}(G_x,G_y,G_z)$, where 
\begin{align*}
F^x&=u_x(u_{yy}-u_{zz})+u_yu_{xy}-u_zu_{xz}\,,\\ F^y&=u_y(u_{zz}-u_{xx})+u_zu_{yz}-u_xu_{xy}\,,\\
F^z&=u_z(u_{xx}-u_{yy})+u_xu_{xz}-u_yu_{yz}\,. 
\end{align*}
Begin by using Procedure \ref{inversion} with $y\prec z\prec x$. In two iterations, this gives
\[
F^x=D_y(u_xu_y)+D_z(-u_xu_z).
\]
With $(x,y,z)$ replacing the indices $(1,2,3)$ in \eqref{Curlinv}, let
\[
H^{xy}=u_xu_y=-H^{yx},\qquad H^{xz}=u_xu_z=-H^{zx}.
\]
Therefore
\[
D_zH^{yz}=F^{y}-D_xH^{yx}=u_yu_{zz}+u_zu_{yz}\,.
\]
One could invert this by a further iteration with the ranking $z\prec y$, though it is inverted more quickly by the line integral formula, which gives $H^{yz}=u_yu_z$. Finally,
\[
G_x=u_yu_z,\qquad G_y=u_xu_z,\qquad G_z=u_xu_y.
\]
\eex

\section*{Acknowledgments} 

I am grateful to Willy Hereman for discussions on homotopy operators and many helpful comments on a draft of this paper, and to the Centre International de Rencontres Math\'{e}matiques in Luminy for support and hospitality during the conference \textit{Symmetry and Computation}. I would like to thank the Isaac Newton Institute for Mathematical Sciences for support and hospitality during the programme \textit{Geometry, Compatibility and Structure Preservation in Computational Differential Equations}; this provided the opportunity to develop much of the theory described above. I also thank the referees for their constructive comments.

\appendix
\section{Appendix}

Procedure \ref{inversion} has been tested on the following conservation laws. In every case, the procedure coupled with the ranking heuristic yields a minimal inversion.

\subsubsection*{Kadomtsev--Petviashvili (KP) equation}

There are two forms of the KP equation, depending on which of $\epsilon=\pm 1$ is chosen. In either case, it has a conservation law,
\begin{align*}
	\mcC&=f(t)y\left(u_{xt}+uu_{xx}+u_x^2+u_{4x}+\epsilon u_{yy}\right)\\
	&=D_x\!\left\{fy\left(u_t+uu_x+u_{xxx}\right)\right\}+D_t\!\left\{\epsilon f(yu_u-u)\right\}.
\end{align*}
With the ranking $x\prec y\prec t$, the procedure requires three iterations to obtain this inversion.

\subsubsection*{Potential Burgers equation (see Wolf \textit{et al.} \cite{WolfBM})}

The potential Burgers equation has a conservation law for each $f(x,t)$ such that $f_t+f_{xx}=0$:
\[
\mcC=f\exp(u/2)(u_t-u_{xx}-\tfrac{1}{2}u^2)=D_x\{\exp(u/2)(2f_x-fu_x)\}+D_t\{2\exp(u/2)f\}.
\] 
The inversion requires one iteration with $x\prec t$, exchanging $f_t$ and $-f_{xx}$ twice.

\subsubsection*{Zakharov--Kuznetzov equation (see Poole \& Hereman \cite{Poole11})}

The Zakharov--Kuznetzov equation is $\mcA=0$, where
\[
\mcA=u_t+uu_x+u_{xxx}+u_{xyy}\,.
\]
It has a conservation law
\begin{align*}
\mcC&=(u^2+2(u_{xx}+u_{yy}))\mcA\\
&=D_x\{(\tfrac{1}{2}u^2+u_{xx}+u_{yy})^2+2u_xu_t\}+D_y\{2u_yu_t\}+D_t\{\tfrac{1}{3}u^3-u_x^2-u_y^2\}.
\end{align*}
The fully-expanded form has just eleven terms and the inversion requires two iterations.
 
\subsubsection*{Short-pulse equation (see Brunelli \cite{Brunelli})}

The short-pulse equation, $\mcA=0$, with
\[
\mcA=u_{xt}-u-uu_x^2-\tfrac{1}{2}u^2u_{xx}\,,
\]
has the following conservation law involving a square root:
\begin{align*}
\mcC&=u_x(1+u_x^2)^{-1/2}\mcA\\
&=D_x\{-\tfrac{1}{2}u^2(1+u_x^2)^{1/2}\}+D_t\{(1+u_x^2)^{1/2}\}.
\end{align*}
With $x\prec t$, this can be inverted in one iteration.

\subsubsection*{Nonlinear Schr\"{o}dinger equation}
Splitting the field into its real and imaginary parts gives the system $\mcA_1=0,\ \mcA_2=0$, with
\[
\mcA_1=-v_t+u_{xx}+(u^2+v^2)u,\qquad\mcA_2=u_t+v_{xx}+(u^2+v^2)v.
\]
One of the conservation laws is
\begin{align*}
	\mcC&=u_t\mcA_1+v_t\mcA_2\\
	&=D_x\!\left\{2u_xu_t+2v_xv_t\right\}+D_t\!\left\{\tfrac{1}{2}\left(u^2+v^2\right)-u_x^2-v_x^2\right\}.
\end{align*}
With the ranking $u\prec v$ and $x\prec t$, the procedure requires two iterations.

\subsubsection*{It\^{o} equations (see Wolf \cite{Wolf})}
The equations are $\mcA_1=0,\ \mcA_2=0$, with
\[
\mcA_1=u_t-u_{xxx}-6uu_x-2vv_x\,,\qquad\mcA_2=v_t-2u_xv-2uv_x\,.
\]
This system has a rational conservation law,
\begin{align*}
\mcC&=2v^{-1}\mcA_1+v^{-4}\left(vv_{xx}-\tfrac{3}{2}v_x^2-2uv^2\right)\mcA_2\\
&=D_x\!\left\{v^{-3}\left(v_xv_t-2u_{xx}v^2-2u_xvv_x-uv_x^2-4u^2v^2-4v^4\right)\right\}+D_t\!\left\{v^{-3}\left(2uv^2-\tfrac{1}{2}v_x^2\right)\right\}.
\end{align*}
With the ranking $u\prec v$ and $x\prec t$, the procedure requires two iterations.

\subsubsection*{Navier--Stokes equations}
The (constant-density) two-dimensional Navier--Stokes equations are $\mcA_\ell=0,\ \ell=1,2,3$, where
\[
\mcA_1=u_t+uu_x+vu_y+p_x-\nu(u_{xx}+u_{yy}),\quad \mcA_2=v_t+uv_x+vv_y+p_y-\nu(v_{xx}+v_{yy}),\quad \mcA_3=u_x+v_y\,.
\]
This system has a family of conservation laws involving two arbitrary functions, $f(t)$ and $g(t)$, namely
\begin{align*}
	\mcC&=f\mcA_1+g\mcA_2+(fu+gv-f'x-g'y)\mcA_3\\
	&=D_x\{fu^2+guv+fp-\nu(fu_x+gv_x)-(f'x+g'y)u\}\\
	&\quad+D_y\{fuv+gv^2+gp-\nu(fu_y+gv_y)-(f'x+g'y)v\}+D_t\{fu+gv\}.
\end{align*}
Procedure C consists of a linear inversion and two further iterations, using the ranking $x\prec y\prec t$. The three-dimensional Navier--Stokes equations have a similar conservation law, which requires a linear inversion and three further iterations.

\bigskip
Procedure C and the ranking heuristic have also been tested on some divergences not arising from conservation laws, with high order or complexity. Again, the output in each case is a minimal inversion.

\subsubsection*{High-order derivatives}
The divergence is
\begin{align*}
	\mcC&=u_{5x}u_{4y}+u_{2x,2y}u_{2x,3y}\\
	&=D_x\!\left\{u_{4x}u_{4y}-u_{3x}u_{x,4y}+(u_{2x,2y})^2 \right\}+D_t\!\left\{u_{3x}u_{2x,3y}-u_{3x,y}u_{2x,2y} \right\}.
\end{align*}
With the ranking $x\prec y$, the procedure requires three iterations. The other ranking, $y\prec x$, also produces a minimal inversion after three iterations; it is higher-order in $x$ but lower-order in $y$:
\[
\mcC=D_x\!\left\{u_{4x,2y}u_{2y}-u_{3x,2y}u_{x,2y}+(u_{2x,2y})^2 \right\}+D_t\!\left\{u_{5x}u_{3y}-u_{5x,y}u_{2y} \right\}.
\]

\subsubsection*{High-order derivatives and explicit dependence}

The divergence is
\begin{align*}
	\mcC&=t(u_yu_{xttt}-u_xu_{yttt})\\
	&=D_x\!\left\{-tuu_{yttt} \right\}+D_y\!\left\{tuu_{xttt}) \right\}.
\end{align*}
With the ranking $x\prec y\prec t$, the procedure requires one iteration. This illustrates the value of the second criterion for ranking independent variables; if $t$ is ranked lower than $x$ and $y$, the procedure fails at the first check.

\subsubsection*{Exponential dependence}
The divergence is
\begin{align*}
	\mcC&=\left(u_{xx}u_y^2-2u_{yy}\right)\exp(u_x)\\
	&=D_x\!\left\{u_y^2 \exp(u_x)\right\}+D_t\!\left\{-2u_y \exp(u_x)\right\}.
\end{align*}
With the ranking $x\prec y$, the procedure requires one iteration.

\subsubsection*{Trigonometric dependence}
The divergence is
\begin{align*}
	\mcC&=\left(u_{xx}u_{yyy}-u_{xxy}u_{xyy}\right)\cos(u_x)\\
	&=D_x\!\left\{u_{yyy} \sin(u_x)\right\}+D_t\!\left\{-u_{xyy} \sin(u_x)\right\}.
\end{align*}
With the ranking $x\prec y$, the procedure requires one iteration.

\end{document}